\documentclass[11pt]{article} 
 
\usepackage{amssymb}
\usepackage{amsmath}

\usepackage[T1]{fontenc}
\usepackage{charter,eulervm}
\usepackage{setspace}

\usepackage{sectsty}

\usepackage{a4}
\usepackage{amsthm}
\usepackage{tikz}  
\usetikzlibrary{trees,calc}
\usepackage{microtype}
\usepackage{xcolor}
\usepackage{nicefrac}
\usepackage[round]{natbib}
\usepackage{stackengine}

\newcounter{comment}
\newcommand{\comment}{\refstepcounter{comment}\noindent\textbf{Comment \arabic{comment}.}~~}

\newcounter{example}
\newcommand{\example}{\refstepcounter{example}\noindent\textbf{Example \arabic{example}.}~~}

\usepackage{hyperref}

\definecolor{dark-red}{rgb}{0.4,0.15,0.15}
\definecolor{dark-blue}{rgb}{0.15,0.15,0.4}
\definecolor{medium-blue}{rgb}{0,0,0.5}
\hypersetup{
    colorlinks, linkcolor={dark-red},
    citecolor={dark-red}, urlcolor={dark-red}
}

\newtheorem{theorem}{Theorem}

\newtheorem{lemma}{Lemma}

\newtheorem{proposition}{Proposition}

\title{Sequential Equilibria in a Class of Infinite Extensive Form Games\thanks{We are grateful to Rabah Amir, Pierpaolo Battigalli, Philippe Bich, Alessandro Bonatti, Gui Carmona, Laura Doval, János Flesch, Christoph Kuzmics, Rida Laraki, George Mailath, Chantal Marlats, Christina Pawlowitsch,  Jérôme Renault, Phil Reny, Klaus Ritzberger, Larry Samuelson, Tristan Tomala, and Nicolas Vieille for valuable discussion and comments.}}
\author{Michael Greinecker, Martin Meier, Konrad Podczeck\thanks{Greinecker: CEPS, ENS Paris-Saclay,
  \href{mailto:michael.greinecker@ens-paris-saclay.fr}{michael.greinecker@ens-paris-saclay.fr}, Meier: Free University of Bozen-Bolzano, \href{mailto:Martin.Meier@unibz.it}{Martin.Meier@unibz.it}, Podczeck: University of Vienna,
  \href{mailto:konrad.podczeck@univie.ac.at}{konrad.podczeck@univie.ac.at}.
  }}

\begin{document} 
\maketitle
\begin{abstract}
Sequential equilibrium is one of the most fundamental refinements of Nash equilibrium for games in extensive form. However, it is not defined for extensive-form games in which a player can choose among a continuum of actions. We define a class of infinite extensive form games in which information behaves continuously as a function of past actions and define a natural notion of sequential equilibrium for this class. Sequential equilibria exist in this class and refine Nash equilibria. In standard finite extensive-form games, our definition selects the same strategy profiles as the traditional notion of sequential equilibrium.
\end{abstract}

\section{Introduction}

Sequential equilibrium is the leading formulation of the idea that an equilibrium should satisfy backward induction, that a player's strategy ought to prescribe optimal play from each conceivable contingency onward. Sequential equilibria were originally defined by \citet{KrepsWilson1982} in terms of limits of completely mixed strategies, behavior strategies under which each action available at an information set is played with positive probability. If uncountably many actions are available at an information set, it is impossible to assign positive probability to all of them; sequential equilibrium never gets off the ground. But many, if not most, applications of game theory in economics require players to choose from a continuum of actions, be it prices, quantities, or any other choice that is naturally modeled as being continuous. Sequential equilibrium cannot be applied here, and researchers are forced to make up ad-hoc refinements that are not supported by any general theory.

We show in this paper that sequential equilibrium admits a natural formulation in a class of infinite extensive form games in which information is suitably continuous. Without such informational continuity assumptions, otherwise seemingly well-behaved games need not even admit a Nash equilibrium. But if we make such an informational continuity assumption, we also get conceptual leverage. Sequential equilibrium derives its power by making off-equilibrium behavior relevant; everything strategically relevant that could happen is considered. Continuity guarantees that everything strategically relevant can be represented by a dense subset. 

Our approach builds on recent work of \citet{MyersonReny2020}. Their starting point is a characterization of sequential equilibrium (behavior) strategy profiles in finite games of perfect recall in which all chance moves have full support: Such a strategy profile is a sequential equilibrium strategy profile if and only if it is the limit of a sequence of profiles of completely mixed strategies in which every player $\epsilon$-optimizes conditional on reaching each information set as $\epsilon$ goes to zero; \citet[Theorem 6.4]{MyersonReny2020}. To adapt their characterization to games with a continuum of actions, Myerson and Reny replace completely mixed strategies with generalized sequences, nets, of strategies that eventually assign positive probability to every action. Their resulting solution concept is given in terms of finitely additive distributions over outcomes as $\epsilon$ goes to zero. These distributions need not be induced by any actual strategy profile. In the class of games that Myerson and Reny study, this is unavoidable. We look at a class of games with nicer continuity properties and obtain better behaved solutions. 

Extensive form games with uncountable action spaces (in discrete time) are haunted by a myriad of pathologies even when payoff and action spaces satisfy standard regularity conditions: Finite-horizon games with observable actions need not have subgame perfect equilibria, \citet*{HarrisRenyRobson1995}, signaling games need not have sequentially rational equilibria, \citet{vanDamme1987} and \citet{Manelli1996}, and even finite-horizon dynamic programming problems with imperfect information need not have solutions, \citet*{HarrisStinchcombeZame2000}.  In our view, these pathologies can be traced to a single source: informational discontinuities.\footnote{Consequently, a continuum of terminal action choices poses no problems\textemdash a fact exploited by \citet{BattigalliTebaldi2019}.} Information in seemingly continuous games can vary discontinuously. In the following example, an informational discontinuity prevents the very existence of any Nash equilibrium, let alone any refinement thereof. Moreover, the space of feasible distributions over plays is not closed in the topology of weak convergence of measures. Here is a very simple example, a souped-up version of matching pennies:\bigskip

\example First, Ann chooses a number $a_1\in [0,1]$ while Bob simultaneously chooses $b\in\{0,1\}$. Afterwards, Ann observes the product $a_1\cdot b$ and chooses $a_2\in\{0,1\}$. Ann's payoff is given by $u_A(a_1,b,a_2)=-a_1-|a_2-b|$ and Bob's payoff by $u_B(a_1,b,a_2)=|a_2-b|$. By choosing $a_1>0$, Ann can guarantee a payoff of  $-a_1$ by matching Bob's action afterwards. However, there can be no equilibrium in which Ann chooses a positive $a_1$ with positive probability since Ann would get a higher payoff by choosing a smaller yet still positive action in the first period. However, if Ann chooses $0$ for sure, Bob can guarantee himself an expected payoff of $1/2$ by randomizing uniformly, which would lead to Ann getting an expected payoff of $-1/2$. Ann could do better by choosing some $a_1<1/2$ initially and then matching Bob for a payoff of $-a_1$. Therefore, no Nash equilibrium exists. 

To see that the space of feasible distributions over plays is not closed, consider a sequence of strategy profiles in which Bob always randomizes uniformly and in which Ann plays first $1/n$ in the $n$th strategy profile, and in which Ann then always mimics Bob's action. For every $n$, the choices $b$ and $a_2$ will be the same and uniformly distributed, and this will still hold in the limit distribution. But in the limit distribution, Ann chooses $0$ in the first period, and this means that she will lack the information to correlate her behavior with his in the limit.
\bigskip

The following example of a signaling game, essentially due to \citet{vanDamme1987} and \citet{Manelli1996}, shows a more subtle informational discontinuity in which information is only used as a correlation device.\bigskip

\example First, nature chooses Ann's type $t$ from $\{-1,1\}$ with equal probability. Then, Ann learns her type and chooses $a\in[-1,1]$. Bob learns $a$ and chooses $b\in\{-1,1\}$. Ann's payoff is given by $u_A(t,a,b)=-|t-b+a|$ and Bob's payoff by $u_B(t,a,b)=a\cdot b$. Bob does not care about Ann's type directly. This game has no sequentially rational Nash equilibrium.\footnote{However, \citet{vanDamme1987} shows that there exists a Nash equilibrium in which Ann never plays an action close to $0$ and for which Bob would react suboptimally to such actions.} To see this, first note that for $a\neq 0$, Bob has a unique best response that matches the sign of $a$, while Bob is indifferent between all of his actions if $a=0$. Suppose there were a sequentially rational equilibrium. If $t=1$, Ann prefers every positive action to every negative action, and her payoff is strictly decreasing in the positive action she chooses towards $0$. The only possible way to play optimally would be to play $a=0$ and to receive a payoff of $0$. By a similar argument, type $-1$ of Ann has to play $a=0$ and receives a payoff of $0$. Since Bob's action is independent of Ann's type, her expected payoff will be strictly negative---a contradiction. 

It is also easy to see that the space of feasible distributions over plays is not closed. Just note that Ann can receive a payoff arbitrarily close to $0$ and her payoff function is continuous. However, she cannot receive a payoff of $0$ under any feasible distributions over plays. 
\bigskip

The phenomenon that the space of distributions over plays need not be closed comes from an effect that Myerson and Reny call strategic entanglement: Future players may use small but informative observable variation as a coordination device. Their behavior is, consequently, correlated. If the original variation goes to zero, the correlation is preserved in the limit without any correlation device inducing it. One encounters this problem  already in the, arguably, simplest class of extensive form games with a continuum of strategies, in games of incomplete information; \citet[Example 2]{MilgromWeber1985}. Games of incomplete information with a continuum of types and even finitely many actions may not have any equilibrium or even an approximate equilibrium; see \citet{Simon2003} and \citet{Hellman2014}. \citet{MilgromWeber1985} introduced an informational continuity assumption for games of incomplete information, which they called absolutely continuous information, and showed that it guarantees existence. \citet{Balder1988} kept the assumption of absolutely continuous information but generalized everything else using the mathematical machinery of Young measures. The absolute continuity assumption of \citet{MilgromWeber1985} makes it possible to analyze a game of incomplete information with correlated types as a game of incomplete information with independent types. The trick is to incorporate density functions for type distributions into the payoff functions. Similarly, \citet{Witsenhausen1988} introduced a parallel assumption in a model of distributed control to reduce a dynamic model in which new information arrives over time to a static model. More recently, \citet{FabbriMoroni2024} introduced a similar dynamic version of the absolute continuity of information assumption and showed that it guarantees the existence of perfect and sequential equilibria in stochastic games with infinitely many types and signals in which a player has only countably many actions available at each information set.  We adapt a version of their informational continuity assumption from a previous version of their paper to allow for a continuum of available actions. As a consequence, we can define a notion of sequential equilibrium that is induced by actual measurable strategies. Distributions over plays are, therefore, countably additive. Countable additivity is not just a technical convenience; finitely additive probability measures are extremely nonconstructive objects and a solution based on them is not explicitly usable. No explicit example of a finitely additive probability measure on a measurable space that is not already countably additive can be constructed.\footnote{If one weakens the axiom of choice to the axiom of dependent choice, a weak form of the axiom of choice that suffices for essentially all real analysis in, say, separable Banach spaces, it is relatively consistent that there are no finitely additive probability measures defined on a $\sigma$-algebra that are not already countably additive; see \citet[Chapter 15]{TomkowiczWagon2016} for the details, in particular Theorem 15.4 and Theorem 15.5.}

To define convergence of strategies, we work with a certain quotient of the space of behavioral strategies that is compact and metrizable, the space of so-called strategic measures---a concept adapted from the literature on nonstationary stochastic dynamic programming which also generalizes the distributional strategies of \citet{MilgromWeber1985}. Strategic measures contain less information than the actual strategies inducing them, and this requires some delicacy in representing continuation behavior.

In general, there need not exist any set of strategy profiles such that every measurable set of private histories that might occur with positive probability must be reached with positive probability, even under our continuity assumption. However,  every measurable set of private histories that has open action-sections and that can be reached with positive probability must be reached with positive probability under a profile of full support strategies (all nonempty open sets of actions are played with positive probability). Since expected payoffs vary continuously in actions in our class of games, such sets allow for meaningful tests of sequential $\epsilon$-rationality. We call these sets ``strategically relevant.'' We then define a sequential equilibrium to be the limit of a convergent sequence of strategy profiles (actually, their induced strategic measures), in which every strategically relevant set is reached with positive probability and such that conditional behavior at each strategically relevant set is eventually $\epsilon$-optimal for each $\epsilon>0$. 
\bigskip

Here is a brief tour of the remainder of the paper: Section \ref{environment} introduces the class of games we study. Section \ref{stratmeas} establishes geometric and topological properties of strategic measures. Section \ref{main} formulates the remaining ingredients for the definition of a sequential equilibrium and the result that the set of sequential equilibria is compact. Section \ref{duop} illustrates our notion of sequential equilibrium by applying it to a sequential duopoly propblem with noisy observability. Section \ref{discuss} discusses some of our assumptions and modeling choices in more detail.
Section \ref{proofs} collects most proofs, except for a few central proofs that we keep in the main text.

\section{The Environment}\label{environment}

Time is discrete with an infinite horizon: $0,1,2\ldots$ There is a nonempty countable (and potentially infinite) set $N$ of \emph{players}, including a player $0$, whom we interpret as nature. All other players are \emph{proper} players. For each period $\tau$, there is a single player $N(\tau)$ active at $\tau$. We assume that $N(\tau)=0$ if and only if $\tau=0$.  We assume without loss of generality that for every player $i$ there is some period $\tau$ such that $N(\tau)=i$.  For every period $\tau$, there is a Polish (separable and completely metrizable) \emph{action space} $A_\tau$. Nature's initial moves are given by $\nu\in\Delta(A_0)$.  Here, we write $\Delta(X)$ for the space of Borel probability measures on a Polish space $X$, endowed with the topology of weak convergence of measures. Unless otherwise mentioned, measurability is understood to be with respect to Borel $\sigma$-algebras. For every nonzero period $\tau$, there is a Polish \emph{signal space} $S_\tau$. For notational convenience, it will be convenient to treat nature's action space as a signal space and write $A_0=S_0$. For every nonzero period $\tau$, we let $H_\tau=\prod_{t=0}^{\tau-1} S_t\times \prod_{t=1}^{\tau-1}A_t$ be the space of \emph{histories} until $\tau$. For nonzero periods $\tau$, there is a measurable \emph{signal function} $\sigma_\tau:H_\tau\to\Delta(S_\tau)$. For every nonzero period $\tau$, we let $H_\tau^p$ be the space of \emph{private histories} of player $N(\tau)$ given by
 \[H_\tau^p=\prod_{\substack{t<\tau \\ N(t)=N(\tau)}} S_t\times\prod_{\substack{t<\tau \\ N(t)=N(\tau)}} A_t.\] For every nonzero period $\tau$, there is a measurable \emph{action correspondence} with nonempty values $\mathcal{A}_\tau:H_\tau^p\times S_\tau\to 2^{A_\tau}$. Here, $H_\tau^p\times S_\tau$ corresponds to an information set at which the player active at $\tau$ moves. We let $P=\prod_{t=0}^\infty S_t\times\prod_{t=1}^\infty A_t$ be the space of (conceivable) \emph{plays}.\footnote{Not every play is feasible given the constraints imposed by the action correspondences.}  Finally, we specify for every proper player $i$ a bounded measurable \emph{payoff function} $v_i:P\to\mathbb{R}$. Note that we allow signals to be payoff relevant, and not all uncertainty that is not due to players randomizing is incorporated in the initial move of nature.\bigskip

 The game proceeds as follows: In period $0$, nature chooses an action $a_0$ in $A_0$ according to $\nu$. Then, player $N(1)$ receives a signal $s_1\in S_1$ distributed according to $\sigma_1(a_0)$ and chooses an action $a_1\in \mathcal{A}_1(s_1)$. Then, player $N(2)$ receives a signal $s_2\in S_2$ distributed according to $\sigma_2(a_0,a_1,s_1)$. If $N(2)$ is a new player, $N(2)\neq N(1)$, then $N(2)$ chooses an action $a_2\in \mathcal{A}_2(s_2)$. If $N(2)=N(1)$, then  $a_2\in \mathcal{A}_2(a_1,s_1,s_2)$. We continue this way. An active player can make their action choices not just dependent on the current signal but also on their complete private history. In particular, we do not interpret signals as information sets containing all available information; signals represent the arrival of new information. That players can condition their choices on their whole private histories guarantees that they have effectively perfect recall. This formulation with signals makes it easier to formulate our informational assumption below, and separates the rules of the game from a player's ability to remember the past.\footnote{\citet{BattigalliGeneroso2024} argue this point forcefully.} That action correspondences depend only on private histories and current signals ensures that the availability of actions provides a player with no additional information.
\bigskip 

\noindent We make the following assumptions on the games we study: Every action correspondence has compact values and is jointly continuous, that is, both upper and lower hemicontinuous, in all action coordinates. All payoff functions are jointly continuous in all action coordinates. Note that we formally require payoff functions to be defined on plays that are not feasible for the given action correspondences. The behavior of payoff functions on infeasible plays is irrelevant to our analysis, and one can always extend payoff functions that are only defined on feasible plays by a suitable extension theorem to all conceivable plays; see \citet{Yushkevich1997}.\bigskip

Central to everything in this paper is the following assumption on the signal functions: For every nonzero period $\tau$, there is a probability measure $\mu_\tau\in\Delta(S_\tau)$ and a measurable function $g_\tau:H_\tau\times S_\tau\to\mathbb{R}$ jointly continuous in all action coordinates in $H_\tau$ such that for all $h\in H_\tau$ and every measurable set $E\subseteq S_\tau$, one has
\[\sigma_\tau(h,E)=\int_E g_\tau(h,s)~\mathrm d\mu_\tau.\] Moreover, we assume that there is a nonnegative measurable function $f_\tau:S_\tau\to\mathbb{R}$ with finite $\mu_\tau$-integral such that $g_\tau(h,s)\leq f_\tau(s)$ for all $h\in H_\tau^p$ and $s\in S_\tau$. In words, signal functions are given by an integrably bounded family of conditional densities that varies continuously in actions. This assumption extends the ``noisy observability'' assumption of a previous version of \citet{FabbriMoroni2024} to games with a continuum of available actions at some moves. The assumption excludes games of perfect information or observable actions when players can choose from a continuum of actions.\footnote{For each perfectly observed action, a different point of positive mass of $\mu_\tau$ is required. Additivity implies that $\mu_\tau$ can have at most countably many such mass points.} For notational convenience, we let $\mu_0=\nu$.\bigskip

A \emph{strategy} $\beta^i$ for a proper player $i$ specifies for each period $\tau$ at which $i$ is active a measurable function $\beta_\tau:H_\tau^p\times S_\tau\to\Delta(A_\tau)$ such that $\beta(h,s)\big(\mathcal{A}_\tau(h,s)\big)=1$ for all $h\in H_\tau^p$ and $s\in S_\tau$. A strategy is \emph{pure}
if all values are point masses (Dirac measures.) Pure strategies can be identified with measurable selections from action correspondences. They exist by the Kuratowski--Ryll-Nardzewski selection theorem; \citet[Theorem 18.13]{AliprantisBorder2003}. A \emph{strategy profile} specifies a strategy for each player. Each strategy profile induces a canonical distribution on $P$ in the obvious way by the Ionescu-Tulcea theorem; see \citet[Theorem 8.24]{Kallenberg2021}.  

\section{Strategic Measures}\label{stratmeas}

Instead of working with strategies directly, it will be convenient to work with somewhat better-behaved derived objects that still allow us to describe the behavior of players while having the geometric and topological structure for us to carry out a fixed-point argument.

Let $\beta_i=(\beta_\tau)_{P(\tau)=i}$ be a strategy of player $i$. Each strategy $\beta_i$ induces recursively a \emph{strategic measure} $\beta_i^*$, a Borel probability measure on $\prod_{N(\tau)=i}S_\tau\times A_\tau$ in which the signal in period $\tau$ is chosen independently from previous period's signals and actions according to $\mu_\tau$, and in which the action in period $\tau$ is chosen according to $\beta_\tau$ and the marginal distribution on $H_\tau^p\times S_\tau$. The condition that signals are distributed independently of their past guarantees, by symmetry, that action choices cannot depend on yet unobserved signals. Profiles of strategic measures can be combined with the densities on signal functions to obtain the actual distribution on plays induced by the strategies but do not directly give the probabilities with which signals and actions are played.\bigskip

\example We illustrate the idea of strategic measures with a simple parametric family of finite extensive form games with our friends Ann and Bob. Ann moves first, and chooses $L$ or $R$. Then, Bob receives a noisy signal with the value $l$ or $r$. If Ann has chosen $L$, Bob receives the signal $l$ with a probability of $c$. If Ann has chosen $R$, Bob receives the signal $r$ with a probability of $c$. We let $w=1-c$. Here $c$ and $w$, represent the probability of a correct or wrong, respectively, signal.  Afterwards, Bob chooses $L$ or $R$ himself. If we take the special case that $c=1$ and, therefore, $w=0$, the signal collapses and we have effectively a game of perfect information, but this is only one extreme case. As to the payoffs, Bob wants to match Ann but Ann wants Bob to mismatch her; the payoff difference is larger when Ann plays $L$. Because both players only move once and signals are payoff-irrelevant, one can identify signals with information sets. The corresponding extensive form looks as follows:\bigskip

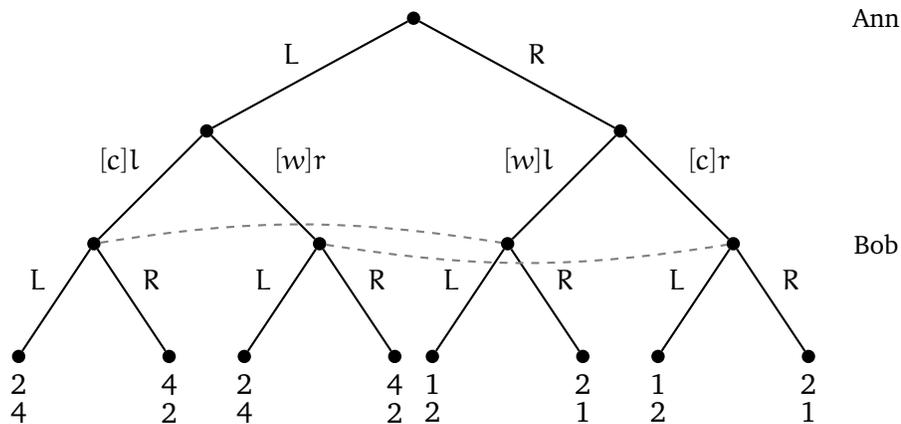
\begin{tikzpicture}
        [
        every path/.style={thick},
        round node/.style={fill=black,circle,draw,inner sep=1.5},
        level distance=15mm,
        level 1/.style={sibling distance=55mm},
        level 2/.style={sibling distance=30mm},
        level 3/.style={sibling distance=20mm},
        ]
        
        \node[round node] (level1) {}
        child {node[round node] (level2) {}
            child {node[round node] (level3) {}
                child {node[round node, label=below:{\stackanchor{$2$}{$4$}}] {} edge from parent node[above left] {$L$}}
                child {node[round node, label=below:{\stackanchor{$4$}{$2$}}] {} edge from parent node[above right] {$R$}}
                edge from parent node[above left] {$[c]l$}
            }
            child {node[round node] (level32) {}
                child {node[round node, label=below:{\stackanchor{$2$}{$4$}}] {} edge from parent node[above left] {$L$}}
                child {node[round node, label=below:{\stackanchor{$4$}{$2$}}] {} edge from parent node[above right] {$R$}}
                edge from parent node[above right] {$[w]r$}
            }
            edge from parent node[above left] {$L$}
        }
        child {node[round node] {}
            child {node[round node] (level33) {}
                child {node[round node, label=below:{\stackanchor{$1$}{$2$}}] {} edge from parent node[above left] {$L$}}
                child {node[round node, label=below:{\stackanchor{$2$}{$1$}}] {} edge from parent node[above right] {$R$}}
                edge from parent node[above left] {$[w]l$}
            }
            child {node[round node] (level34) {}
                child {node[round node, label=below:{\stackanchor{$1$}{$2$}}] {} edge from parent node[above left] {$L$}}
                child {node[round node, label=below:{\stackanchor{$2$}{$1$}}] (lowright) {} edge from parent node[above right] {$R$}}
                edge from parent node[above right] {$[c]r$}
            }
            edge from parent node[above right] {$R$}
        };
        
        \coordinate (legendX) at ($(lowright)+(.9,0)$);
        \node at (legendX|-level1) {\small Ann};
        \node at (legendX|-level2) {};
        \node at (legendX|-level3) {\small Bob};
        
        \path[draw, dashed, bend left=9, gray] (level3) to (level33);
        \path[draw, dashed, bend right=9, gray] (level32) to (level34);
    \end{tikzpicture}\bigskip

\noindent For Ann, a strategic measure is simply a probability distribution over her actions; her behavior need not condition on any information. For Bob, we have to take some probability distribution on signals such that real signals come with a positive density. Every probability distribution that assigns positive probability to both signals works; we choose $\mu$ so that both $l$ and $r$ are chosen uniformly with probability $1/2$ each. To make this work with the system of conditional densities $g$, we need that $g(L,l)\cdot 1/2=g(R,r)\cdot 1/2=c$ and $g(L,r)\cdot 1/2=g(R,l)\cdot 1/2=w$. Therefore, $g(L,l)=g(R,r)=2c$ and $g(L,r)=g(R,l)=2w$. A strategic measure for Bob is just a probability distribution on $\{l,r\}\times\{L,R\}$ with uniform $\{l,r\}$-marginal. A behavior strategy of Bob is a system of conditional probabilities for each signal, and the corresponding strategic measure must have these as the actual conditional probabilities. For example, if $m$ is a strategic measure of Bob, the corresponding behavior strategy plays $L$ with probability $m(l,L)/\big(m(l,L)+m(l,R)\big)= 2 m(l,L)$ after Bob sees the signal $l$. Let $m_A$ be a strategic measure of Ann, which is simply a probability distribution over $\{L,R\}$. For Ann, there is no difference between a mixed strategy, a behavior strategy, and a strategic measure. Let $m_B$ be a strategic measure of Bob. Let $(a_1,s_2,a_2)$ be a play. The probability  that this play occurs under behavior strategies giving rise to the strategic measures $m_A$ and $m_B$, respectively, is $m_A(a_1)\cdot g(a_1,s_2)\cdot m_B(s_2,a_2)$. To see that this is the correct distribution, observe that the true probability of the signal $s_2$ occurring after action $a_1$ is $g(a_1,s_2)\cdot 1/2$ and the conditional probability of Bob playing $a_2$ after $s_2$ is $2\cdot m_B(s_2,a_2)$. Note that the distribution over plays is multilinear in strategic measures. 

One can transform the game into an equivalent game without informative signals but with different payoff functions. In this transformed game, signals are uninformative and arrive with probability $\mu$, independently of what Ann does. The payoffs are the original payoffs, multiplied by the density,.  which makes signals payoff relevant. This gives rise to the following extensive form game:
\bigskip

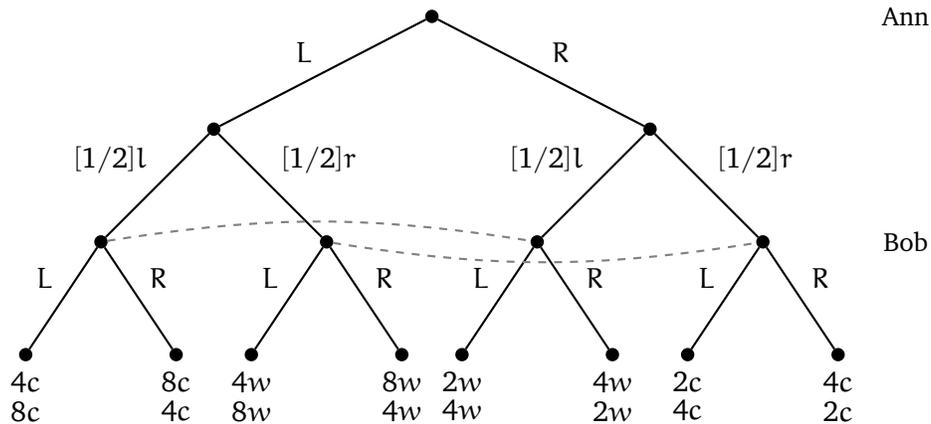
\begin{tikzpicture}
        [
        every path/.style={thick},
        round node/.style={fill=black,circle,draw,inner sep=1.5},
        level distance=15mm,
        level 1/.style={sibling distance=58mm},
        level 2/.style={sibling distance=30mm},
        level 3/.style={sibling distance=20mm},
        ]
        
        \node[round node] (level1) {}
        child {node[round node] (level2) {}
            child {node[round node] (level3) {}
                child {node[round node, label=below:{\stackanchor{$4c$}{$8c$}}] {} edge from parent node[above left] {$L$}}
                child {node[round node, label=below:{\stackanchor{$8c$}{$4c$}}] {} edge from parent node[above right] {$R$}}
                edge from parent node[above left] {$[1/2]l$}
            }
            child {node[round node] (level32) {}
                child {node[round node, label=below:{\stackanchor{$4w$}{$8w$}}] {} edge from parent node[above left] {$L$}}
                child {node[round node, label=below:{\stackanchor{$8w$}{$4w$}}] {} edge from parent node[above right] {$R$}}
                edge from parent node[above right] {$[1/2]r$}
            }
            edge from parent node[above left] {$L$}
        }
        child {node[round node] {}
            child {node[round node] (level33) {}
                child {node[round node, label=below:{\stackanchor{$2w$}{$4w$}}] {} edge from parent node[above left] {$L$}}
                child {node[round node, label=below:{\stackanchor{$4w$}{$2w$}}] {} edge from parent node[above right] {$R$}}
                edge from parent node[above left] {$[1/2]l$}
            }
            child {node[round node] (level34) {}
                child {node[round node, label=below:{\stackanchor{$2c$}{$4c$}}] {} edge from parent node[above left] {$L$}}
                child {node[round node, label=below:{\stackanchor{$4c$}{$2c$}}] (lowright) {} edge from parent node[above right] {$R$}}
                edge from parent node[above right] {$[1/2]r$}
            }
            edge from parent node[above right] {$R$}
        };
        
        \coordinate (legendX) at ($(lowright)+(.9,0)$);
        \node at (legendX|-level1) {\small Ann};
        \node at (legendX|-level2) {};
        \node at (legendX|-level3) {\small Bob};
        
        \path[draw, dashed, bend left=9, gray] (level3) to (level33);
        \path[draw, dashed, bend right=9, gray] (level32) to (level34);
    \end{tikzpicture}\bigskip

In this game, Ann and Bob act independently and can not condition their behavior at each other. For this reason, we get a game in normal form in which players choose strategic measures. Since distributions over plays are multlinear in strategic measures, so are expected payoffs. 
As we will see, all these properties generalize.\bigskip

\comment While strategic measures contain enough information in our setting to formulate notions such as Nash equilibrium, some information gets lost when going from strategies to strategic measures. For one, two strategies of player $i$ that agree $\mu_\tau$-almost surely for every period $\tau$ at which $i$ is active must necessarily induce the same strategic measure. More importantly, strategic measures do not specify what a player does after playing actions that will never be played under this strategic measure. \citet{Kuhn2003} observed in 1952 that normal strategies contain more information than is needed to describe a player's contingent behavior in any situation the player may find themselves in when following the strategy. A strategic measure is closer to what \citet{Rubinstein1991} called a ``plan of action'' than to an ordinary strategy. Importantly, a strategic measure will not specify behavior at every information set, and this makes it harder to formulate an appropriate notion of sequential rationality. It is also worth pointing out that strategic measures are closer to mixed strategies than to behavior strategies. Indeed, it follows from \citet[Theorem 5.2]{Feinberg1996} that every strategic measure is a mixture of strategic measures induced by pure strategies.\footnote{Together with Lemma \ref{lchtvs} below, the fact that the weak topology on a Polish space is metrizable, and with Choquet's theorem (see \citet[Chapter 3]{Phelps2001}), Feinberg's theorem shows that the extreme points of the set of strategic measures of a player are exactly the strategic measures induced by pure strategies.} Also, as Lemma \ref{expected} below shows, expected payoffs are linear in a player's strategic measures, a property that would not hold for behavior strategies.\bigskip

\example We illustrate the issues arising when a player moves more than once in a simple, finite dynamic decision problem with incomplete information. The state of nature is either fine or bad; both cases happen with equal probability. The decision-maker can first decide whether to learn the state or remain uninformed. Afterwards, the decision-maker invests or does not invest. For the phenomena we illustrate, a specification of payoffs is not needed. The decision tree looks as follows:

\bigskip

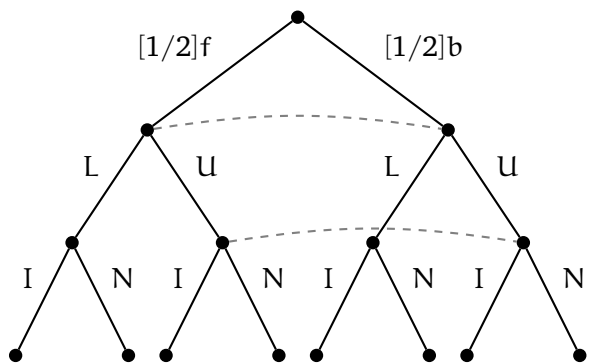
\begin{tikzpicture}
        [
        every path/.style={thick},
        round node/.style={fill=black,circle,draw,inner sep=1.5},
        level distance=15mm,
        level 1/.style={sibling distance=40mm},
        level 2/.style={sibling distance=20mm},
        level 3/.style={sibling distance=15mm},
        ]
        
        \node[round node] (level1) {}
        child {node[round node] (level2) {}
            child {node[round node] (level3) {}
                child {node[round node] {} edge from parent node[above left] {$I$}}
                child {node[round node] {} edge from parent node[above right] {$N$}}
                edge from parent node[above left] {$L$}
            }
            child {node[round node] (level32) {}
                child {node[round node] {} edge from parent node[above left] {$I$}}
                child {node[round node] {} edge from parent node[above right] {$N$}}
                edge from parent node[above right] {$U$}
            }
            edge from parent node[above left] {$[1/2]f$}
        }
        child {node[round node] (level22) {}
            child {node[round node] (level33) {}
                child {node[round node] {} edge from parent node[above left] {$I$}}
                child {node[round node] {} edge from parent node[above right] {$N$}}
                edge from parent node[above left] {$L$}
            }
            child {node[round node] (level34) {}
                child {node[round node] {} edge from parent node[above left] {$I$}}
                child {node[round node] (lowright) {} edge from parent node[above right] {$N$}}
                edge from parent node[above right] {$U$}
            }
            edge from parent node[above right] {$[1/2]b$}
        };
%
%
        \path[draw, dashed, bend left=9, gray] (level2) to (level22);
        \path[draw, dashed, bend left=9, gray] (level32) to (level34);
    \end{tikzpicture}\bigskip

\noindent Note that when the decision-maker moves a second time, they are either uninformed of the state, know it is fine, or know it is bad. We, therefore, need three signals, $\{u,f,b\}$. As the underlying probability measure, we choose the uniform distribution assigning probability $1/3$ to each of these signals. A strategic measure of the decision-maker is then a probability distribution over $\{L,U\}\times\{u,f,b\}\times\{I,N\}$ with the property that the $\{u,f,b\}$-marginal is uniform and independent of the $\{L,U\}$-marginal. For the density, we have $g(f,U,u)=g(b,U,u)=g(f,L,f)=g(b,L,b)=3$, and all other values are $0$. If $m_0$ is it the strategic measure of nature in the initial period and $m$ the strategic measure of the decision maker, the probability of the play $(a_0,a_1,s_2,a_2)$ is, given by $m_0(a_0)m(a_1,s_2,a_3)g(a_0,a_1,s_2)$. Then, to take an example, the conditional probability of not investing of a decision maker after playing $U$ and receiving the signal $f$ is $m(U,f,N)/\big(m(U,f,N) + m(U,f,I)\big)$. This probability is only defined if the decision maker actually plays $U$ with positive probability, while a behavior strategy specifies such a conditional probability even when the probability of playing $U$ is $0$. This is the kind of information one loses when going from behavior strategies to strategic measures.\bigskip

Strategic measures have a nice geometric and topological structure. If action correspondences are constant, the following lemma is a special case of a result of \citet{Nowak1988}. Similar to \citet{Balder1989}, we show that our space of strategic measures can be viewed as a closed subspace of a space that is compact by the result of \citet{Nowak1988}. Again, $\Delta(X)$ denotes the space of Borel probability measures on a Polish space $X$. Recall that the \emph{weak topology} on $\Delta(X)$ is the topology generated by mappings of the form $\mu\mapsto\int f~\mathrm d\mu$ with $f:X\to\mathbb{R}$ bounded and continuous.\footnote{We have no need to refer to weak topologies on normed spaces in this paper; no confusion should arise.}

\begin{lemma}\label{compconv}The set of strategic measures induced by the strategies of a player is a nonempty, convex, and compact set in the weak topology.
\end{lemma}

\begin{lemma}\label{lchtvs}The set of strategic measures induced by the strategies of a player when endowed with the weak topology is homeomorphic to a nonempty compact convex subset of a locally convex Hausdorff topological vector space.
\end{lemma}

We write $M_i\subseteq\Delta\big(\prod_{N(\tau)=i}S_\tau\times A_\tau\big)$ for the space of strategic measures induced by strategies of proper player $i$, and let $M_0$ be the trivial space containing only the single measure $\nu$. A generic element of $M_i$ is usually denoted by $m_i$.  Before we are able to show, at least, that Nash equilibria exist by a fixed-point argument in the space of profiles of strategic measures, we need to establish the continuity of expected payoffs. It is easier to do this when the time horizon is finite, something we can actually assume without loss of generality. The following lemma shows that our assumptions guarantee that a stochastic version of the continuity at infinity assumption of \citet{FudenbergLevine1983} holds.

\begin{lemma}\label{continf} For each player $i\in N$ and each $\epsilon>0$, there exists a period $T$ and a payoff function $v_i^T$ depending only on the first $T$ periods such that the expected payoff under $v_i$ and $v_i^T$ differs by at most $\epsilon$ for each feasible distribution over plays. 
\end{lemma}

Using Lemma \ref{continf}, we are able to establish that expected payoffs are jointly continuous and multilinear in the space of strategic measures. Together with Lemma \ref{lchtvs} and the Kakutani–Fan–Glicksberg fixed-point theorem,  \citet[Corollary 17.55]{AliprantisBorder2003}, this guarantees at least the existence of a Nash equilibrium in our environment.

\begin{lemma}\label{expected}The expected payoff of a player under two strategy profiles that induce the same profile of strategic measures is the same. The expected payoff is then a continuous multilinear function on the space of profiles of strategic measures.
\end{lemma}

The idea of Lemma \ref{expected} is that one can identify suitable independent product measures obtained from strategic measures as distributions on plays in games in which signals do not depend on histories. To get the correct expected payoffs, one multiplies the payoff function from the original game with the corresponding densities.

In order to go beyond Nash equilibrium, we need to define conditional behavior and a substitute for continuation strategies in the setting of strategic measures. Let $m_i$ be a strategic measure of proper player $i\in N$ and let $N(\tau)=i$. A \emph{$\tau$-continuation} of $m_i$ is a strategic measure $m_i'$ whose $\prod_{N(t)=i,t<\tau}S_t\times A_t$-marginal coincides with the corresponding marginal of $m_i$. Since the marginal mapping is continuous and linear, the set of such $\tau$-continuations is compact and convex, too, by Lemma \ref{compconv}:

\begin{lemma}\label{cont}Let $m_i$ be a strategic measure of proper player $i\in N$ and let $N(\tau)=i$. Then the set of $\tau$-continuations of $m_i$ is nonempty, compact, and convex.
\end{lemma}

\comment An issue with strategic measures is that they are defined in terms of the measures $\mu_\tau$. Formally, these are not part of the data of the model; the signal functions are. The assumptions on signal functions guarantee that such measures exist but do not single out any canonical measure. It turns out that the choice is immaterial. One can show that there exist such measures with the property that they assign probability zero to events that cannot have positive probability in any possible distribution over plays; this follows from Lemma \ref{fullsupp} and Lemma \ref{fullsuppinfcompl} below by taking suitable marginals. For any such measures, the corresponding spaces of strategic measures are linearly homeomorphic, and the induced topology on equivalence classes of actual strategies is the same.

\section{Sequential Equilibrium}\label{main}

To motivate our definition of a sequential equilibrium, we revisit the characterization \citet{MyersonReny2020} give of sequential equilibrium strategy profiles in finite games of perfect recall in which all chance moves have full support.\footnote{ For other characterizations of sequential equilibria in such finite games that do not refer to belief systems, see \citet[Proposition 6]{KrepsWilson1982} and \citet[Proposition B]{blume1994algebraic}, and, most similar, \citet[]{halpern2009nonstandard}. \citet[Proposition 2.1]{dilme2024sequentially} characterizes the feasible distributions on terminal nodes induced by sequential equilibria.} It turns out that completely mixed strategies play only a minor role in the characterization of \citet{MyersonReny2020}; one can replace completely mixed strategy profiles by strategy profiles in which every information set is reached with positive probability in the characterization (the result is not restricted to multi-stage games):  

\begin{proposition}\label{characterization}
A strategy profile in a finite game with perfect recall in which every move of nature has strictly positive probability is a sequential equilibrium strategy profile if and only if it is the limit of a sequence of profiles of strategies under which every information set occurs with positive probability and in which every player $\epsilon$-optimizes conditional on reaching each information set as $\epsilon$ goes to zero.
\end{proposition}
\begin{proof}
That a sequential equilibrium strategy profile satisfies the condition follows directly from \citet[Theorem 6.4]{MyersonReny2020}. For the other direction, take a strategy profile that is the limit of a sequence of profiles of strategies under which every information set occurs with positive probability and in which every player $\epsilon$-optimizes conditional on reaching each information set as $\epsilon$ goes to zero. Now, conditional expected payoffs are continuous in strategy profiles at information sets that occur with strictly positive probability. Consequently, one can perturb each term in the sequence so that all actions are played with positive probability, such that the limit is unchanged, and such that every player $\epsilon$-optimizes conditional on reaching each information set as $\epsilon$ goes to zero (possibly, slower). The result follows then, again, from \citet[Theorem 6.4]{MyersonReny2020}. 
\end{proof}
 
Completely mixed strategies play a more essential role for the notion of a perfect equilibrium of \citet{Selten1975} or the notion of a quasi-perfect equilibrium of \citet{vanDamme1984}, in which  they ensure (a form of) admissibility. Even well-behaved infinite games with simultaneous moves need not have admissible equilibria.\footnote{For explicit examples, see Example 2.1 of \citet{SimonStinchcombe1995} or Example 3 of \citet*{Mendez-Naya1995}.} 

Still, since information sets (which coincide with private histories in our environment) include actions, it is in general not possible to find a single strategy profile under which every measurable family of information sets that could be reached with positive probability is reached with positive probability. However, it is not necessary to test $\epsilon$-optimality at all such families. Since we have enough continuity in actions, we can test conditional $\epsilon$-optimality on events that include slack in actions. Let $i\in N$ be a proper player active at $\tau$. We call a measurable set (of information sets) $F\subseteq H_\tau^p\times S_\tau$ \emph{strategically relevant} at $\tau$ for $i$ if $F$ is reached with positive probability under some strategy profile and if for each $(s,a)\in F$, there exists an open neighborhood $U$ of $a$ such that $(s,a')\in F$ for all $a'\in U$. Equivalently, $F\subseteq H_\tau^p\times S_\tau$ is strategically relevant at $\tau$ for $i$ if $F$ is reached with positive probability under some strategy profile and if for each $s$ the $s$-section of $F$ given by $\{a\mid (s,a)\in F\}$ is open.

 A measurable set that is strategically relevant at some period for some player is simply \emph{strategically relevant}; we can identify every such set with the set of compatible plays. A \emph{full support strategy} of a player is a strategy that always chooses a full support distribution over available actions, a distribution under which every nonempty open set has positive probability. Such full support strategies do always exist. A strategic measure induced by a full support strategy is a \emph{full support strategic measure}.

\begin{lemma}\label{fullsupp}Every player has a full support strategy.
\end{lemma}

\begin{lemma}\label{fullsuppinfcompl}Under a profile of full support strategic measures, each strategically relevant set is reached with positive probability. 
\end{lemma}
\noindent The converse of Lemma \ref{fullsuppinfcompl} does not hold; see Example \ref{twice} below. 

Let $m$ be a profile of strategic measures, $F$ be strategically relevant at $\tau$ to $i$, and assume that $F$ is reached with positive probability under $m$. Then $m_i$ is \emph{$\epsilon$-optimal} at $F$ against $m_{-i}$ if the conditional expected payoff of $(m_i,m_{-i})$ given $F$ is higher than the  conditional expected payoff of $(m_i',m_{-i})$ given $F$ minus $\epsilon$ for every $\tau$-continuation $m_i'$ of $m_i$. Note that the probability of reaching $F$ is the same under $(m_i,m_{-i})$ and $(m_i',m_{-i})$.\bigskip

A profile of strategic measures $m$ is a \emph{sequential equilibrium} if it is the limit of a convergent sequence $\langle m^n \rangle$ of profiles of strategic measures under which every strategically relevant set is reached with positive probability and and for which for every $\epsilon>0$, all but finitely many $n$, and every proper player $i$, $m_i^n$ is $\epsilon$-optimal against $m_{-i}^n$ at every measurable set strategically relevant at any period $\tau$ to $i$. 

That such limits $m$ correspond to sequential equilibrium strategy profiles in a standard finite game follows from Proposition \ref{characterization}, with the caveat that a strategic measure does not prescribe behavior at information sets that cannot occur under that strategic measure. The corresponding beliefs are then obtained by taking any limit point of the induced beliefs by $\langle m^n \rangle$. Clearly, $m$ must be a Nash equilibrium.

We are now ready to prove the existence of a sequential equilibrium. Since strategic measures are closer to mixed strategies than behavior strategies, a bit of care is needed to ensure conditional $\epsilon$-optimality. Also, we need a little lemma about infinite Minkowski sums.

\begin{lemma}\label{Mink}Let $K$ be a compact subset of a Hausdorff topological vector space, $\langle C_n\rangle$ a sequence of compact, convex subsets of $K$, and $\langle \alpha_n\rangle$ a sequence of nonnegative numbers such that $\sum_n \alpha_n=1$. Then the set
\[\sum_n \alpha_n C_n=\bigg\{\sum_n \alpha_n x_n\Bigm\vert x_n\in C_n\textnormal{ for all }n\bigg\}\]
is well-defined, compact, and convex. 
\end{lemma}

\begin{theorem}\label{exist}A sequential equilibrium exists.
\end{theorem}
\begin{proof}By Lemma \ref{fullsupp}, there exists a profile of full support strategic measures $(m_i^*)_{i\in N}$. For each proper player $i$ and each natural number $n\geq 2$, we construct a modified space of strategic measures $M_i^n$ of player $i$ such that for corresponding profiles of strategic measures every strategically relevant set is reached with positive probability and such that best-responses in such spaces guarantee conditional $1/n$-optimality.  For each period $\tau$ at which $i$ is active, let $M_i^*(\tau)$ be the space of $\tau$-continuations of $m_i^*$. Now, define a modified space of strategic measures by
\[M_i^n=\sum_{N(\tau)=i} 1/n^\tau M_i^*(\tau)+\bigg(1-\sum_{N(\tau)=i} 1/n^\tau\bigg)M_i.\]
By Lemma \ref{compconv}, Lemma \ref{lchtvs}, and Lemma \ref{Mink}, $M_i^n$ is compact and convex as a weighted average of compact and convex sets and, since it contains $m_i^*$, nonempty. Moreover, for a profile of strategic measures in which each proper player $i$ chooses from $M_i^n$, every strategically relevant set is reached with positive probability.  Combining everything with Lemma \ref{expected} and the Kakutani–Fan–Glicksberg fixed-point theorem, there exists for each $n\geq 2$ a Nash equilibrium $m^n$ in which all players are restricted to choosing strategic measures from the restricted space $M_i^n$.  One can think of player $i$ being forced to choose from $M_i^n$ as them being constrained to choosing from $M_i^*(\tau)$ with probability $1/n^\tau$ for each $\tau$ such that $N(\tau)=i$, and being otherwise unconstrained. By Lemma \ref{compconv} and passing to a subsequence, we can assume without loss of generality that $\langle m^n\rangle$ converges. 

It remains to show that for every $\epsilon>0$, all but finitely many $n$, every proper player $i\in N$, and every set $F$ strategically relevant in some period $\tau$ to some player $i$, $m_i^n$ is $\epsilon$-optimal at $F$ against $m_{-i}^n$. By Lemma \ref{fullsuppinfcompl}, $F$ is reached with positive probability under each $m^n$. Assume without loss of generality that $v_i$ takes values in $[0,1]$. Let $n>1/\epsilon$. We can decompose $m_i^n$ as 
\[m_i^n=\sum_{N(t)=i} 1/n^t m_t+\bigg(1-\sum_{N(t)=i} 1/n^t\bigg)m\]
with $m_t\in M_i^*(t)$ for each $t$ and $m\in M_i$. Note that if $F$ is reached with positive probability under $m$ or $m_t$ with $t\leq \tau$, then the corresponding $\tau$-continuation must already be optimal since the respective spaces of strategic measures they belong to impose no restrictions after period $\tau$. All conditional suboptimal behavior must come from the $m_t$ with $t>\tau$. Note also that the probability of reaching $F$ is the same (and positive) for all $m_t$ with $t\geq\tau$. It suffices, therefore, to show that the conditional probability of $F$ being reached by $m_t$ with $t>\tau$ is small enough. This probability is at most 
\[\frac{\sum_{N(t)=i, t>\tau} 1/n^t}{1/n^\tau+\sum_{N(t)=i, t>\tau}1/n^t}\leq \frac{\sum_{t=\tau+1}^\infty 1/n^t}{\sum_{t=\tau}^\infty 1/n^t}=1/n<\epsilon.\] Here, we have used that $a/b\leq (a+c)/(b+c)$ for $0<a\leq b$ and $c\geq 0$.
Since the payoff function has values in $[0,1]$, the conditional expected payoff must be within $\epsilon$ of the maximal conditional expected payoff.
\end{proof}

\begin{proposition}The set of sequential equilibria is compact.
\end{proposition}
\begin{proof} Since the space of profiles of strategic measures is compact, it suffices to show that the set of sequential equilibria is closed. Let $m$ be a closure point of the set of sequential equilibria. Let $d$ be a metric that metrizes the topology on the space of profiles of strategic measures. For each $k$, let $m_k$ be a sequential equilibrium such that $d(m,m_k)<1/k$. For each $n$, let $\langle m_k^n\rangle$ be a sequence converging to $m_k$ that verifies that $m_k$ is a sequential equilibrium. Let $n(0)$ be any natural number and define recursively for all $k>0$ the number $n(k)$ as the smallest number larger than $n(k-1)$ such that for all $n\geq n(k)$ every player plays  $1/k$-optimal in $m_k^n$ against the other players at every measurable set strategically relevant at any period to the player, and such that $d(m_k^n,m_k)<1/k$. Then, $d(m,m_k^{n(k)})<2/k$, so the sequence $\langle m_k^{n(k)}\rangle$ converges to $m$. By construction, the sequence verifies that $m$ is a sequential equilibrium.
\end{proof}

In the proof of Theorem \ref{exist}, we made use of the fact that the mixture of a profile of strategic measures under which every strategically relevant set is reached with positive probability with another profile of strategic measure is a profile of strategic measures under which every strategically relevant set is reached with positive probability. In particular, the mixture of a profile of full suppport strategic measure s with another profile of strategic measure is a profile of strategic measures under which every strategically relevant set is reached with positive probability. However a mixture of a full support strategic measure with another strategic measure need not be a full support strategic measure:\bigskip

\example\label{twice} Consider a single player who has to choose twice from $[0,1]$. One full support strategy chooses first uniformly from $[0,1]$ and chooses again uniformly from $[0,1]$ after each action in the first period. Consider now the strategy that first chooses $0$ and in the second period chooses the previous action. Take a mixture of the corresponding strategic measures with probability $1/2$. For the resulting strategic measure, $0$ is chosen with probability $1/2$, and the conditional probability of choosing $0$ in the second period given the choice of $0$ in the first period is $1$. So a strategy inducing this strategic measure cannot have full support after playing $0$ in the first period.

\section{Application: Noisy sequential Duopoly}\label{duop}

We apply our solution concept to a noisy version of a sequential duopoly. It is a classic result that when firms compete in quantities 
à la Cournot, there is a first mover advantage when firms move sequentially. The result is based on the second firm to move observing the quantity choice of the first firm perfectly. In a famous paper, \citet{bagwell1995commitment} argued that perfect observability is too strong an assumption for realistic applications and suggested that the second firm to move should only receive a noisy signal of the first firm's chosen quantity. In a setting in which firms can only choose among finitely many quantities and there are finitely many possible signals, Bagwell showed that if the support of the conditional signal distributions does not vary with the actions of the first firm, the only equilibrium quantity choices in any Nash equilibrium in pure strategies are the static Cournot quantities. This result has been adapted to a setting with continuous quantities by \citet{Maggi1999}. Maggi actually allows the supports of signals to vary, and we build on his result. In his setting, two identical firms choose quantities from an interval $[0,\bar{q}]$ with $\bar{q}$ larger than the quantity a single firm acting as a monopolist would choose. The profit of firm producing the quantity $q_i$ against $q_{-i}$ is given by $\pi(q_i,q_{-i})$. The profit function $\pi$ is assumed to be twice continuously differentiable and satisfies the derivative conditions $\pi_{11}<\pi_{12}<0$. To exclude boundary solutions, we assume that the marginal profit is strictly positive at $0$ when a single firm acts as a monopolist. The strict concavity in own quantities implies that best responses are unique and pure; the best responses are given by the reaction function $q_{-i}\mapsto R(q_{-i})$. There exists a unique Nash equilibrium, which is symmetric and in pure strategies. In this equilibrium, both firms choose the \emph{Cournot quantity} $q^C$. To see this, one uses implicit differentiation to show that the reaction function is differentiable and satisfies $-1<R'<0$. By the mean value theorem, we must have for two quantities $x<y$ some $c\in(x,y)$ such that $R(x)-R(y)=R'(c)(x-y)$. Therefore, \[|R(x)-R(y)|\leq |R'(c)||x-y|<|x-y|,\] so an equilibrium must be symmetric; and there can clearly be at most one symmetric equilibrium. If firms move sequentially under perfect observability, there is a unique subgame perfect equilibrium. The first firm to choose, the leader, chooses a quantity $q^l$ that maximizes $\pi(q_i,R(q_i))$, and the second firm, the follower, chooses $q^f=R(q^l)$. Then $\pi(q^l,q^f)>\pi(q^C,q^C)>\pi(q^f,q^l)$; the leader has an advantage. 

In Maggi's model with imperfect observability,\footnote{\citet{Maggi1999} extends the model to allow for incomplete information, too, and that is the main focus of his paper. Here, we focus on the model with complete information but imperfect observability.} the signal is given by a jointly continuous conditional density $g:[0,\bar{q}]\times\mathbb{R}\to\mathbb{R}$ such that for two nondecreasing differentiable functions $g_0,g_1:[0,\bar{q}]\to\mathbb{R}$, one has $g_0(q)<q<g_1(q)$ for all $q\in [0,\bar{q}]$, and such that $g(q,x)>0$ if and only if $g_0(q)<x<g_1(q)$. This model fits all of our assumptions.\footnote{To see that the conditional densities are integrably bounded, note that $g$ vanishes outside $[0,\bar{q}]\times[q_0(0),q_1(\bar{q})]$ and is bounded there.} Now, \citet[Proposition 1]{Maggi1999} says that at every Nash equilibrium in pure strategies, the leader plays $q^C$, and the follower chooses $q^C$ (almost surely) on the equilibrium path. However, as Maggi also notes, such an equilibrium must sometimes violate a minimal form of sequential rationality: $q^C$ need not be a be a best reply to a belief supported on $[q_0(q),q_1(q)]$ for all $q\in [0,\bar{q}]$. We show that if $g_1-g_0$ is small enough, there is a first mover advantage. A sequential equilibrium exists by Theorem \ref{exist}. We go through the proof in the main text to illustrate how our notion of sequential equilibrium can be practically applied.
\begin{proposition}Fix $\pi$. For all $\epsilon>0$, there exists $\delta>0$ such that if the signal satisfies $|g_1(q)-g_0(q)|\leq \delta$ for all $q\in [0,\bar{q}]$, then the expected payoff of the leader in every sequential equilibrium must be at least $\pi(q^l,q^f)-\epsilon$. 
\end{proposition}
\begin{proof} Let $\delta=\max_q |g_1(q)-g_0(q)|$. If every strategically relevant set is reached with positive probability, every interval of the form $[q_0(q),q_1(q)]$ for some $q\in [0,\bar{q}]$ must be reached with positive probability. The only way for the follower to receive a signal in $[q_0(q),q_1(q)]$ with positive probability is if the leader played an action in $[q_0(q)-\delta ,q_1(q)+\delta]$. A best-reply to a probability distribution $\tau$ supported on  $[q_0(q)-\delta ,q_1(q)+\delta]$ must be supported on $\big[R\big(q_0(q)-\delta\big),R\big(q_1(q)+\delta\big)\big]$. Indeed, if a point $q'$ does not lie in this interval, it is possible to find a point that is closer to every best reply to a point in $[q_0(q)-\delta ,q_1(q)+\delta]$, and such point must give a higher payoff on average under $\tau$. Even stronger, since $\pi$ is continuously differentiable and strictly concave in a firm's own quantity, the follower's payoff decreases uniformly as they move away from $\big[R\big(q_0(q)-\delta\big),R\big(q_1(q)+\delta\big)\big]$. This implies that for small $\epsilon'>0$, conditional $\epsilon'$-optimal behavior at $[q_0(q),q_1(q)]$ requires playing with high probability an action in a narrow closed interval around $\big[R\big(q_0(q)-\delta\big),R\big(q_1(q)+\delta\big)\big]$. If $\mu$ is a probability measure underlying the signal space for which $g$ is a suitable family of conditional densities, the definition of a sequential equilibrium together with the Portmanteau theorem characterizing weak convergence implies that if $m$ is the strategic measure of the follower in a sequential equilibrium, $m$ must assign measure $\mu\big([q_0(q),q_1(q)]\big)$ to $[q_0(q),q_1(q)]\times \big[R\big(q_0(q)-\delta\big),R\big(q_1(q)+\delta\big)\big]$ for all $q\in [0,\bar{q}]$. 

Now, take $\epsilon>0$. Since both $R$ and $\pi$ are uniformly continuous, for $\delta$ small enough, it must be the case that the leader playing $q^l$ must get an expected payoff of at least $\pi(q^l,q^f)-\epsilon$. A fortiori, the leader's expected payoff in a sequential equilibrium must be at least as high.   
\end{proof}
It follows from Maggi's result that for $\delta$ small enough, a sequential equilibrium cannot be in pure strategies. The logic underlying Maggi's result goes as follows: If the leader plays a pure quantity $q^*$ in an equilibrium, the follower must almost surely play the best response quantity $R(q^*)$ for signals  in the interval $[q_0(q^*),q_1(q^*)]$. If the leader changes their quantity slightly towards the Cournot quantity, the interval of possible signals changes slightly at the boundary. But at the boundary, continuity forces the density to be close to zero and the payoff-function is bounded, so the effect on the leader's payoff from a change of this interval is of second order. At the same time, there is a dominating first order effect at the major part of the interval that is unchanged, so this effect dominates. The argument does not apply to mixed equilibria; the follower can play a strategy that is decreasing in the signal received.. 

For the original finite model of \citet{bagwell1995commitment}, in which the support of the signals does not depend on the leader's choice, \citet{VanDammeHurkens1997} show that there is a mixed equilibrium in which the leader has a first mover advantage. In Bagwell's model, sequential equilibrium has no bite whatsoever. However, \citet{VanDammeHurkens1997} argue for an equilibrium selection principle that singles out the mixed equilibrium in which the leader has a first mover advantage. In the setting of \citet{Maggi1999} for $\delta$ small enough, our notion of sequential equilibrium has bite and guarantees a first mover advantage.

\section{Discussion}\label{discuss}

Throughout, we assume that payoff functions are bounded. Since our method of proof uses the fact that the conditional density functions for signals can be included in modified payoff functions, we only need an integrable boundedness condition for the modified payoff functions. Sufficient conditions for this to be possible are provided by \citet{FabbriMoroni2024}. Making the unmodified payoff functions bounded allows for a clean separation between assumptions on payoff functions and assumptions on signal functions. In this paper, we want to put the focus on the latter.

Our formulation of an extensive form game is as a multistage game with an infinite time horizon and no simultaneous moves. By making action choices trivial from some period on, we can, of course, accommodate a finite time horizon, and simultaneous moves can be incorporated in the usual way.\footnote{If one allows for a countable number of simultaneously active players, an additional argument is needed to show that the mapping from strategic measures to distributions over histories of a given length is continuous; use \citet[Corollary 2.4.8]{Bogachev2018}.} A more substantive restriction is that players know when they move. In games of perfect recall, every player must have a subjective linear notion of time, as shown by \citet{Ritzberger1999}, but these subjective notions of time need not coincide with any objective notion of time, as required in our formulation. Our formulation saves us from the tedious accounting task of specifying how profiles of private histories are mapped to objective histories. As long as this mapping is continuous, all arguments go through. Indeed, only the proof of Lemma \ref{expected} needs to be adapted.  

The class of games \citet{FabbriMoroni2024} look at is more general in everything but the restriction to finitely many players, countably many available actions at each information set, and a restriction to additively separable payoffs. In particular, Fabbri and Moroni do not restrict themselves to multistage games. 

The class of games we use is not directly comparable with the class of games for which \citet{MyersonReny2020} prove their existence theorem, the class of regular projective games. An obvious difference is that we allow for a countably infinite number of players, an infinite time horizon, and history-dependent action correspondences. On the other hand, we exclude games of perfect information or observable actions in which a player can choose among a continuum of actions. To look at the regular projective games that also satisfy our assumptions, one would need to look at regular projective games in which players only observe some moves of nature, and no two players observe the same  moves of nature. From this perspective, we allow for payoffs that are discontinuous in the moves of nature (signals) and weaken the assumption R.5. on moves of nature made by \citet{MyersonReny2020}. A full support perfect $\epsilon$-equilibrium, the kind of strategy profile Myerson and Reny show exists, induces a strategic measure that is $\epsilon$-optimal at each relevant set. This follows from them defining their notion of limits using a very fine topolgy that requires a form of uniform convergence. This topology is much finer than weak convergence of induced strategic measures, so every limit in their sense is a limit in ours, and they require pointwise conditional optimality. Myerson and Reny perturb moves of nature in their equilibrium notion in a way that makes it impossible to use strategic measures. This, and the pointwise optimality requirement prevent any direct way to transform a conditionally $\epsilon$-optimal equilibrium in strategic measures into a perfect $\epsilon$-equilibrium.

One pleasant feature of sequential equilibria not captured by our approach are conditional beliefs. \citet{KrepsWilson1982} argued that an advantage of working with assessments (a pair consisting of a profile of behavior strategies and a system of conditional beliefs) is that the plausibility of an equilibrium can often be discussed in terms of the plausibility of beliefs. Incorporating conditional beliefs in our solution concept faces the difficulty that there exists no sufficiently fine topology on conditional beliefs to capture sequential rationality that is also coarse enough to be compact. In the setting of large signaling games, \citet*{PereaJansenPeters1997} showed with an example that even pointwise weak convergence of beliefs might be too weak a topology. \citet{MyersonReny2020} have a sequential $\epsilon$-rationality result with respect to finitely additive conditional beliefs as their Theorem 6.15, but not for limit conditional beliefs.\footnote{In their setting, that would not help since there need not be limit strategies.} Similarly, we could easily construct (even countably additive) conditional beliefs for each $m_n$ in $\langle m^n \rangle$ using regular conditional probabilities, \citet[Theorem 8.5]{Kallenberg2021}, such that $\epsilon$-sequential rationality holds on each relevant set. However, there are no useful limit conditional beliefs for which our solution guarantees exact sequential rationality.

\section{Proofs}\label{proofs}

Unless otherwise specified, we endow the product of topological spaces with the product topology, the product of measurable spaces with the product $\sigma$-algebra, and all topological spaces with their Borel $\sigma$-algebras. In the context we are working with (countable products of Polish spaces), these conventions will never clash. We make heavy use of the topology of weak convergence of measures. Since all our probability measures live on Polish spaces, this topology is metrizable; \citet[Theorem 15.12]{AliprantisBorder2003}. In particular, we can work with sequences instead of nets.

If $(X,\mathcal{X})$ is a measurable space and $Y$ and $Z$ are Polish spaces, a function $h:X\times Y\to Z$ is a \emph{Carathéodory function} if $h(\cdot,y):X\to Z$ is measurable for each $y\in Y$ and $h(x,\cdot):Y\to Z$ is continuous for each $x\in X$. In that case, $h$ is jointly measurable; \citet[Lemma 4.51]{AliprantisBorder2003}. We make frequent use of the following Scorza Dragoni theorem, a version of Lusin's theorem for Carathéodory functions.\footnote{The special  case in which $Y$ is compact does actually follow from Lusin's theorem, \citet[Theorem 12.8]{AliprantisBorder2003}, with the help of \citet[Lemma 3.99 and Theorem 4.55]{AliprantisBorder2003}.}

\begin{lemma}\label{ScorzaDragoni}Let $X,Y,Z$ be Polish spaces, $\nu\in\Delta(X)$, and $h:X\times Y\to Z$ a Carathéodory function. Then there exists for each $\epsilon>0$ a compact set $K\subseteq X$ such that $\nu(X\setminus K)<\epsilon$ and such that the restriction of $h$ to $K\times Y$ is continuous.
\end{lemma}
\begin{proof}See \citet*[Theorem 2.5.19]{Denkowskietal2003}.
\end{proof}

An important consequence of the Scorza-Dragoni theorem is that integrals of certain possibly unbounded and discontinuous functions converge under the topology of weak convergence. While such arguments are well known, see for example \citet{BerliocchiLasry1973}, we give the argument for the reader's convenience. Importantly, in our Polish setting, the weak topology coincides with the more involved weak topology on transition probabilities (``narrow convergence of Young measures'') used by \citet{Balder1988}. Another consequence is that the actual topologies on signal-spaces play essentially no role; only their Borel $\sigma$-algebras matter.

\begin{lemma}\label{Young} Let $X$ and $Y$ be Polish spaces, $\Delta_\nu(X\times Y)$ be the space of probability measures on $X\times Y$ with $X$-marginal $\nu$ endowed with the topology of weak convergence of measures. Let $h:X\times Y\to\mathbb{R}$ be a Carathéodory function such that for some nonnegative $\nu$-integrable function $f:X\to\mathbb{R}$ one has $|h(x,y)|\leq f(x)$ for all $x$ and $y$. Let $\langle \mu_n\rangle$ be a sequence in $\Delta_\nu(X\times Y)$ converging to $\mu$. Then $\lim_n \int h~\mathrm d\mu_n=\int h~\mathrm d\mu$.
\end{lemma}
\begin{proof}Let $\epsilon>0$. Since $f$ is $\nu$-integrable, there is some $M>0$ such that 
\[\int_{\{x\mid f(x)\geq M\}}f~\mathrm d\nu\leq\epsilon.\]
Let $h^*$ be given by $h^*=h\cdot 1_{\{x\mid f(x)< M\}\times Y}$. By the Scorza Dragoni theorem, there exists a compact set $K\subseteq X$ such that $\nu(X\setminus K)<\epsilon/2M$ and such that $h^*$ restricted to $K\times Y$ is continuous. By the Tietze extension theorem, \citet[Theorem 2.47]{AliprantisBorder2003}, there exists a continuous function $h^{**}:X\times Y\to [-M,M]$ that coincides with $h^*$ on $K\times Y$. Note that for all $\lambda\in\Delta_\nu(X\times Y)$, we have
\[\bigg|\int h~\mathrm d\lambda-\int h^{**}~\mathrm d\lambda\bigg|=\bigg|\int h~\mathrm d\lambda-\int h^*\mathrm d\lambda+\int h^*\mathrm d\lambda-\int h^{**}~\mathrm d\lambda\bigg|\]
\[\leq\bigg|\int h~\mathrm d\lambda-\int h^*\mathrm d\lambda\bigg|+\bigg|\int h^*\mathrm d\lambda-\int h^{**}~\mathrm d\lambda\bigg|\leq \int |h-h^*|~\mathrm d\lambda + \int |h^*-h^{**}|~\mathrm d\lambda\]
\[=\int_{\{x\mid f(x)\geq M\}\times Y} |h|~\mathrm d\lambda+\int_{(X\setminus K)\times Y}|h^*-h^{**}|~\mathrm d\lambda\leq \epsilon+\epsilon.\]
Now, by weak convergence of measures, we have for $n$ large enough 
\[\bigg|\int h^{**}~\mathrm d\mu_n-\int h^{**}~\mathrm d\mu\bigg|<\epsilon\] and, therefore,
\[\bigg|\int h~\mathrm d\mu_n-\int h~\mathrm d\mu\bigg|=\]
\[\bigg|\int h~\mathrm d\mu_n-\int h^{**}~\mathrm d \mu_n +\int h^{**}~\mathrm d \mu_n-\int h^{**}~\mathrm d\mu+\int h^{**}~\mathrm d\mu-\int h~\mathrm d\mu\bigg|\]
\[\leq \bigg|\int h~\mathrm d\mu_n-\int h^{**}~\mathrm d \mu_n\bigg| + \bigg|\int h^{**}~\mathrm d \mu_n-\int h^{**}~\mathrm d\mu\bigg|+\bigg|\int h^{**}~\mathrm d\mu-\int h~\mathrm d\mu\bigg|\]
\[\leq 2\epsilon+\epsilon+2\epsilon.\]
\end{proof}

\begin{proof}[Proof of Lemma \ref{compconv}]Nonemptiness follows from the existence of strategies. By \citet[Corollary 3.42]{AliprantisBorder2003}, we can view each $A_\tau$ as a Borel subset of a compact metrizable space $A_\tau^*$. If we allow actions to be chosen in these extended spaces and ignore the constraints provided by action correspondences, compactness follows from \citet[Theorem 2]{Nowak1988}. The proof Nowak uses also shows that strategic measures are determined by linear constraints, so the set of strategic measures is convex. We show that those measures being supported on the action correspondences are compact too. 

Fix a player. By relabeling, we can assume without loss of generality that the player is active in the periods $1,2,3,\ldots$ By \citet[Theorem 18.10 and Theorem 17.15]{AliprantisBorder2003}, we can identify a measurable correspondence whose values are nonempty compact subsets of $A_\tau$ with a measurable function into the space $\mathcal{K}_\tau$ of nonempty compact subsets of the Polish space $A_\tau$ endowed with the Hausdorff metric topology, and a continuous correspondence whose values are nonempty compact subsets of $A_\tau$ as a continuous function into $\mathcal{K}_\tau$. Viewed this way, we can identify $\mathcal{A}_\tau$ with a Carathéodory function \[\phi_\tau:\prod_{t\leq\tau}S_t\times\prod_{t<\tau} A_\tau\to \mathcal{K}_\tau.\]
We now construct for each natural number $m$ recursively a sequence of compact sets $\langle K_{\tau,m}\rangle$ with $K_{\tau,m}\subseteq\prod_{t=1}^\tau S_t$ for each $\tau$. For every nonzero period $\tau$, let 
\[\Gamma_\tau=\big\{(s_1,a_1,\ldots,s_\tau,a_\tau)\mid a_t\in\mathcal{A}_t(s_1,a_1,\ldots,s_{t-1},a_{t-1},s_t)\text{ for }t=1,\ldots,\tau\big\}.\]
By Lemma \ref{ScorzaDragoni}, the Scorza Dragoni theorem, (or just by Lusin's theorem), there exists some compact set $K_{1,m}\subseteq S_1$ such that $\phi_1$ is continuous on $K_{1,m}$ and $\mu_1(S_1\setminus K_{1,m})\leq \nicefrac{1}{2m}$.  By \citet[Theorem 17.15 ]{AliprantisBorder2003}, this implies that the correspondence $\mathcal{A}_1$ is continuous on $K_{1,m}$ and, in particular, upper hemicontinuous. Since upper hemicontinuous correspondences with compact values map compact sets to compact sets, \citet[Lemma 17.8 ]{AliprantisBorder2003}, the set $K_{1,m}\times\mathcal{A}_1(K_{1,m})$ is compact too.  Given that $K_{1,m},\ldots,K_{\tau-1,m}$ are defined, we use the Scorza-Dragoni theorem to construct a compact set $K_{\tau,m}\subseteq S_1\times \cdots\times S_\tau$ 
such that \[\otimes_{t=1}^\tau\mu_t(S_1\times\cdots\times S_\tau\setminus K_{\tau,m})\leq \nicefrac{1}{2^\tau m}\] and such that $\phi_\tau$ is continuous on the compact set
\[\begin{split}
C_m^\tau=\big\{(s_1,a_1,\ldots,s_\tau)\mid~ & (s_1,a_1,\ldots,s_{\tau-1},a_{\tau-1})\in\Gamma_{\tau-1}\textnormal{ and}\\
& (s_1,\ldots,s_t)\in K_{t,m}\text{ for every }t=1,\ldots,\tau\big\}.
\end{split}\]
Now, let $C_m\subseteq\prod_{\tau=1}^\infty S_\tau\times\prod_{\tau=1}^\infty A_\tau$ be the set of elements such that the projection onto $\prod_{t=1}^{\tau} A_t\times\prod_{t=1}^\tau S_t$ lies in $C_m^\tau$ for every $\tau$. Given a sequence with values in $C_m$, one can use a diagonal argument to construct a subsequence such that the projection onto each set of the form $\prod_{t=1}^{\tau} A_t\times\prod_{t=1}^\tau S_t$ converges using the compactness of $C_m^\tau$. So $C_m$ is compact too. Notice that a compact set is compact in every space it is embedded in; the set $C_m$ is still compact if we would replace each $A_\tau$ by $A_\tau^*$. A strategic measure on the larger space that ignores the action correspondences is an actual strategic measure if and only if it assigns measure at least $1-1/m$ to $C_m$ for each $m$. For each $m$, the set of measures that assign measure at least $1-1/m$ to $C_m$ is closed by the Portmanteau theorem. We can, therefore, write the set of strategic measures as an intersection of compact sets of measures, which is again compact in the space of measures defined on the larger space. To finish the proof, we identify the space of measures on $\prod_{\tau=1}^\infty S_\tau\times\prod_{\tau=1}^\infty A_\tau$ with the space of measures on $\prod_{\tau=1}^\infty S_\tau\times\prod_{\tau=1}^\infty A_\tau^*$ supported on the former space. Under this identification, the weak topology on the space of measures on $\prod_{\tau=1}^\infty S_\tau\times\prod_{\tau=1}^\infty A_\tau$ is the trace of the weak topology
 on the space of measures on $\prod_{\tau=1}^\infty S_\tau\times\prod_{\tau=1}^\infty A_\tau^*$ by Lemma \citet[Lemma 15.4]{AliprantisBorder2003}. Consequently, the set of strategic measures of a player is compact in the weak topology.
\end{proof}

\begin{proof}[Proof of Lemma \ref{lchtvs}]The Polish space $\prod_{N(\tau)=i}S_\tau\times \prod_{N(\tau)=i}A_\tau$ is homeomorphic to a Borel (actually, a $G_\delta$) subset of a compact metrizable space $K$ by \citet[Corollary 3.42]{AliprantisBorder2003}. We can then identify a probability measure on the former space with a probability measure on $K$ that is supported on the former space. Clearly, this identification preserves convexity. By Lemma \citet[Lemma 15.4]{AliprantisBorder2003}, the weak topology on $\Delta(\prod_{N(\tau)=i}S_\tau\times A_\tau)$ is, under the identification, the trace of the weak topology on $\Delta(K)$. Since compactness in the trace topology implies compactness in the ambient space, the set of strategic measures can be identified with a compact convex subset of $\Delta(K)$. But  $\Delta(K)$ endowed with the weak topology can be identified with a closed subset of the unit ball of the dual of $C(K)$ endowed with the weak*-topology by the Riesz representation theorem; \citet[Corollary 14.15]{AliprantisBorder2003}, and, therefore, with a compact convex subset of a locally convex Hausdorff topological vector space by Alaoglu's theorem; \citet[Theorem 6.21]{AliprantisBorder2003}.
\end{proof}

\begin{proof}[Proof of Lemma \ref{continf}]Assume without loss of generality that $v_i$ takes values between $0$ and $1$. By iterative application of the Scorza-Dragoni theorem in a way similar to the proof of Lemma \ref{compconv}, we can construct a compact set $C\subseteq P$ on which $v_i$ is continuous and such that with probability at least $1-\epsilon/2$, all feasible plays will lie in $C$. Also as in the proof of Lemma \ref{compconv}, we can think of $P$ as being embedded in a compact product space $\prod_{\tau=0}^\infty S_\tau^*\times\prod_{\tau=1}^\infty A_\tau^*$. Again, $C$ is compact in the larger spaces and, therefore, closed in the larger space. By the Tietze extension theorem, we can extend the restriction of $v_i$ to $C$ to a continuous function on the larger space that takes values between $0$ and $1$ too. To finish the proof, it suffices to show that there exists a continuous function $v_i^T:\prod_{\tau=1}^\infty S_\tau^*\times\prod_{\tau=0}^\infty A_\tau^*\to \mathbb{R}$ that is uniformly $\epsilon/2$-close to this function and depends on only $T$-periods for some $T$. The family of continuous functions on the compact space $\prod_{\tau=1}^\infty S_\tau^*\times\prod_{\tau=0}^\infty A_\tau^*$ that depend on only finitely many periods is a point-separating algebra containing the constants. By the Stone-Weierstrass theorem, \citet[Theorem 9.13]{AliprantisBorder2003}, it is uniformly dense. 
\end{proof}

\begin{proof}[Proof of Lemma \ref{expected}]
It follows from Lemma \ref{continf} that we can assume without loss of generality that payoffs depend only on actions and signals up to some period $T$ in order to prove continuity.  Note that each $(m_i)_{i\in N}\in\prod_{i\in N}M_i$ induces a probability measure $\otimes_{i\in N}m_i$ that is, up to relabeling, a probability measure on $P$ whose $\prod_{\tau=0}^T S_\tau$-marginal is $\otimes_{\tau=0}^T\mu_\tau$ and with all such marginals independent of each other. Moreover, the function from profiles of strategic measures to such a measure on $P$ is continuous by \citet[Theorem 2.8]{Billingsley1999}.  However, to obtain the correct induced probability measure corresponding to the strategies, one has to take account of the densities $g_\tau$. Let each $f_\tau$ be as specified in our continuity of information assumption. Let $f:\prod_{\tau=1}^T S_\tau\to\mathbb{R}$ be defined by $f(s_1,s_2,\ldots,s_T)=f_1(s)f_2(s_2)\cdots f_T(s_T)$. The nonnegative function $f$ is $\otimes_{\tau=1}^T\mu_\tau$ integrable by Tonelli's theorem; \citet[Theorem 11.28]{AliprantisBorder2003}. Write $P_T$ for the space of initial segments of plays up to period $T$. We can interpret each $g_\tau$ as a function on $P_T$ that only depends on $H_\tau\times S_\tau$. We can then define $g:P_T\to\mathbb{R}$ as $g_1g_2\cdots g_T$. The true distribution on $P_T$ with respect to the real strategies has Radon-Nikodym derivative $g_T$ with respect to the corresponding marginal of $\otimes_{i\in N}m_i$.  Now, player $j$'s expected payoff function in terms of strategic measures is given by \[\otimes_{i\in N}m_i\mapsto\int v_jg~\mathrm d\otimes_{i\in N}m_i.\] The continuity of this function follows from Lemma \ref{Young}. To apply the lemma, note that $v_j g$ is a Carathéodory function, as the pointwise product of two  Carathéodory functions. Also, since $v_j$ takes values in $[0,1]$, we have for all $(s,a)\in\prod_{\tau=1}^T S_\tau\times\prod_{\tau=0}^T A_\tau$, that $|v_j g(s,a)|=v_j g(s,a)\leq g(s,a)\leq f(s)$.  

Multilinearity is straightforward by Tonelli's theorem; \citet[Theorem 11.28]{AliprantisBorder2003}.\footnote{Often, a related result on iterated integrals of nonnegative function is known as Tonelli's theorem. Here, it refers to the fact that iterated integrals are equal if the individual functions are integrable under $\sigma$-finite measures.}
\end{proof}

\begin{proof}[Proof of Lemma \ref{fullsupp}]
For each correspondence $\mathcal{A}_\tau$ there exists a countable sequence $\langle f_n\rangle$ of measurable functions such that \[\mathcal{A}_\tau(h,s)=\mathnormal{cl}\Big(\big\{f_n(h,s)\mid n\in\mathbb{N}\big\}\Big)\]
by the Castaing representation of a measurable correspondence with nonempty closed values in a Polish space; \citet[Corollary 18.14]{AliprantisBorder2003}. One can then construct full support strategies with $\beta_\tau$ given by 
\[\beta_\tau(h,s)=\sum_{n=1}^\infty 1/2^l \delta_{f_n(h,s)}\]
for all $(h,s)\in H_\tau^p\times S_\tau$.
\end{proof}

\begin{lemma}\label{productify}Let $S$ and $A$ be Polish spaces. Let $F\subseteq S\times A$ be a measurable set such that whenever $(s,a)$ is in $F$, then there is an open neighborhood $U$ of $a$ such that $\{s\}\times U\subseteq F$. Let $\mu$ be a probability measure on $S\times A$ such that $\mu(F)>0$. Then there exists a measurable set $E\subseteq S$ and an open set $O\subseteq A$ such that $E\times O\subseteq F$ and $\mu(E\times O)>0$.
\end{lemma}

\begin{proof}Let $\mathcal{B}$ be a countable base for the topology on $A$. For each $O\in\mathcal{B}$, let \[F_O=\big\{s\in S\mid \{s\}\times O\subseteq F\big\}=\pi_S\Big(S\times O\cap F^C\Big)^C,\]
with $\pi_S:S\times A\to S$ being the projection onto the first coordinate.
By construction, $F_O\times O\subseteq F$, $F=\bigcup_{O\in\mathcal{B}}F_O\times O$, and $F_O$ is universally measurable as a co-analytic set;  \citet[Theorem 12.24 and Theorem 12.41]{AliprantisBorder2003}. The product $\sigma$-algebra obtained from the universally measurable sets in $S$ and the Borel sets of $A$ is coarser than the $\sigma$-algebra of universally measurable subsets of $S\times A$, and $\mu$ has a natural extension to the latter. So $\mu(F_O\times O)$ is well-defined. Since $F=\bigcup_{O\in\mathcal{B}}F_O\times O$ and $\mu(F)>0$, the countable subadditivity of $\mu$ implies that $\mu(F_O\times O)>0$ for some $O\in\mathcal{B}$. There must also exist a Borel set $E$ and a null set $N\subseteq S$ under the $S$-marginal of $\mu$ such that $F_O=E\cup N$. Then $E\times O$ gives us the desired set.
\end{proof}

\begin{proof}[Proof of Lemma \ref{fullsuppinfcompl}]
We can view a set $F$ that is strategically relevant to $i$ at $\tau$ as a subset $F_1$ of
\[A_0\times \prod_{t=1}^\tau S_t\times\prod_{t=1}^{\tau-1}A_t.\] 
Fix some probability distribution over plays induced by a strategy profiles under which $F_1$ is reached with positive probability. By Lemma \ref{productify}, the fact that each open set is a countable union of open rectangles, and the countably subadditivity of measures, there exists a measurable set $S^1\subseteq\prod_{t=0}^\tau S_t$ and and open rectangle $O_1^1\times\cdots\times O_{\tau-1}^1\subseteq\prod_{t=1}^{\tau-1}A_t$ such that $S^1\times O_1^1\times\cdots\times O_{\tau-1}^1\subseteq F_1$ and such that the set is reached with positive probability too. Now, the set
\[F_2=\bigg\{(s,a)\in A_0\times \prod_{t=1}^\tau S_t\times\prod_{t=1}^{\tau-2}A_t~\bigg\vert~\mathcal{A}_{\tau-1}(s,a)\cap O^1_{\tau-1}\neq\emptyset\bigg\}\]
must be reached with positive probability. It is measurable because $\mathcal{A}_{\tau-1}$ is measurable. Also, because $\mathcal{A}_{\tau-1}(s,\cdot)$ is lower hemicontinuous for each $s$, the set $F_2$ has the property that for each $(s,a)\in F_2$, there exists an open neighborhood $U$ of $a$ such that $(s,a')\in F_2$ for all $a'\in U$. Using again Lemma \ref{productify}, the fact that each open set is a countable union of open rectangles, and the countably subadditivity of measures, there exists a measurable set $S^2\subseteq A_0\times \prod_{t=1}^{\tau} S_t$ and and open rectangle $O_1^2\times\cdots\times O_{\tau-2}^2\subseteq\prod_{t=1}^{\tau-2}A_t$ such that $S^2\times O_1^2\times\cdots\times O_{\tau-2}^2\subseteq F_2$ and such that the set is reached with positive probability. Continuing this way, we recursively construct and define $S^{n-1}$, $O_1^{n-1},\cdots, O_{\tau-n+1}^{n-1}$ and
\[F_n=\bigg\{(s,a)\in A_0\times \prod_{t=1}^\tau S_t\times\prod_{t=1}^{\tau-n}A_t~\bigg\vert~\mathcal{A}_{\tau-n+1}(s,a)\cap O^{n-1}_{\tau-n+1}\neq\emptyset\bigg\}\]
such that $F_n$ is reached with positive probability. Then
\[S^\tau\times\prod_{n=1}^{\tau-1}O_{\tau-n}^n\]
is reached with positive probability under some strategy. Also, by construction, it will be reached with positive probability under each full support strategy profile.
\end{proof}

\begin{proof}[Proof of Lemma \ref{Mink}] If some $C_n$ is empty, there is nothing to show. So we can exclude that case. By Tychonoff's theorem, \citet[Theorem 2.61]{AliprantisBorder2003}, the nonempty set $\prod_n C_n$ is compact, and it is clearly convex in the natural linear space. It suffices to show that the function from $S:\prod_n C_n\to K$ given by $S(x_1,x_2,\ldots)=\sum_n\alpha_n x_n$ is well-defined and continuous. $S$ is then obviously also linear, so the image is nonempty, compact, and convex. 
We first show that $S(x_1,x_2,\ldots)$ is well defined. Indeed, let $K_N=\sum_{n=1}^N\alpha_n x_n+(1-\sum_{n=1}^N\alpha_n)K$. Then $\langle K_N\rangle$ is a nested sequence of compact sets with the finite intersection property, and the intersection contains, therefore, a point $x$. Let $x'$ be a point in this intersection, let $U$ be a neighborhood of the origin, and let $U'$ be a symmetric neighborhood of the origin such that $U'+U'+U'\subseteq U$. Since $K$ is bounded as a compact set, we have $(1-\sum_{n=1}^N\alpha_n)K\subseteq U'$ for $N$ large enough and, thus, $x-x'\in U'+U'\subseteq U$. Since $U$ was arbitrary, the Hausdorff property implies $x=x'$. So $S(x_1,x_2,\ldots)=x$ and $S$ is a well-defined mapping. For all $(y_1,y_2,\ldots)$ in a suitable neighborhood of $(x_1,x_2,\ldots)$, we have $\sum_{n=1}^N \alpha_n x_n-\sum_{n=1}^N \alpha_n y_n\in U'$ by the continuity of the vector operations. Therefore,  $\sum_n\alpha_n x_n-\sum_n\alpha_n y_n\in U'+U'+U'\subseteq U$. So $S$ is continuous. 
\end{proof}

\pdfbookmark{References}{References}
\small

\end{document}